\documentclass[10pt]{article}
\usepackage[ruled,vlined]{algorithm2e}
\usepackage{epsfig}
\usepackage{amsthm}

\newtheorem{theorem}{Theorem}
\newtheorem{lemma}{Lemma}
\newtheorem{definition}{Definition}
\newtheorem{corollary}{Corollary}
\begin{document}
\date{}

%
%

\title{Eccentricity of the nodes of OTIS-cube and Enhanced-OTIS-cube}

\author{Rajib K Das
\\
Department of Computer Science \& Engineering, University of Calcutta\\
West Bengal, India}

\maketitle

\begin{abstract}In this paper we have classified the nodes of OTIS-cube based on their
eccentricities. OTIS (optical transpose interconnection system) is
a large scale optoelectronic computer architecture,
proposed in \cite{KMKE92},
that benefit from both optical and electronic
technologies. We show that radius and diameter of OTIS-$Q_n$ is $n+1$ and
$2n+1$ respectively. We also show that average eccentricity of OTIS-cube
is $(3n/2+1)$.
 In \cite{D05},  a variant of OTIS-cube,  called Enhanced OTIS-cube (E-OTIS-$Q_n$) was proposed.
 E-OTIS-$Q_n$  is regular of degree $n+1$ and maximally fault-tolerant.
  In this paper we have
given a classification of the nodes of  E-OTIS cube and derived expressions
for the eccentricities of the nodes in each class. Based on these results we
show that radius and diameter of E-OTIS-$Q_n$ is $n+1$ and
$\lfloor {4n+4/3} \rfloor$ respectively. We have also
computed the average eccentricity of E-OTIS-$Q_n$ for values of $n$ upto
20.
\end{abstract}

\section{Introduction}

The optical transpose
interconnection system (OTIS), first introduced in
\cite{KMKE92},
was proposed for large scale parallel computer architectures in
which the processors are divided into groups. Each group of
processors is fabricated on one or more high density chips
or modules with in-built electronic interprocessor connections, whereas processors in different groups are interconnected via free space optics. When the distance between modules is more than
a few millimeters, the free space optical interconnects
offer power, speed, I/O bandwidth and crosstalk advantages over electronic
counterparts, and hence are suitable for {\it large scale systems} [8].  Since it has been shown in \cite{KMKE92} that both the bandwidth
and power consumption are minimized when the number of
processors in a group is equal to the number of groups, most of the topologies proposed so far, based on
OTIS,  follows this design guideline \cite{RS98}, \cite{SW97}, \cite{ZMPE00}.
In such an OTIS system, an intergroup link connects processor $p$ of
group $g$ to processor $g$ of group $p$. The
interconnection topology formed by the intra-group links are called here the  {\it factor} network. Depending
on the factor network, different OTIS-network topologies have been proposed so far, such as
OTIS-Mesh, OTIS-Cube, OTIS-star and so on.

A general study of OTIS-Networks based on any factor network topology,
is presented
in \cite{DA02}. In \cite{D04}, a number of results regarding the topological
properties of OTIS-$k$-ary $n$-cube interconnection networks have been
derived which shows the suitability of these topologies for
multiprocessor interconnection network.
In a more recent work \cite{CXP09}, the authors  contribute further results in this
direction by studying a general fault tolerance property of
OTIS networks. It is proved that an OTIS network is
maximally fault tolerant if its basis network is connected. They have also
proposed a corresponding method for constructing
parallel paths between its nodes.
Extensive research has been done to develop efficient algorithms for various applications on the OTIS architecture such as
  selection and sorting \cite{RS98},
matrix multiplication \cite{WS01}, image processing \cite{WS00}, BPC permutation generation \cite{SW97} etc.

OTIS-cube is a particular member of the general class of OTIS networks where
the factor network is a hypercube. An OTIS-cube $Q_n$ consists of $2^n$ groups
with each group having $2^n$ nodes. It is to be noted that in each group of an OTIS-cube
there is one node whose processor number is the same as the group number. So, there
is no optical link from such  a node. Hence the OTIS-cube (OTIS-$Q_n$) is not a regular topology, since
in this network some nodes are of degree $n+1$ and some nodes are of
degree $n$. In \cite {CXP09}, the authors raised the question whether the fault-tolerance
of OTIS structure can be further improved by pairwise connecting the nodes which
do not possess inter-cluster links.
The topology
 proposed in \cite{D05},
called the  enhanced OTIS-cube (E-OTIS-$Q_n$), does just that. There, it is shown
that if the Basis network is hypercube, the fault-tolerance
is improved by adding optical links to the nodes which are of degree $n$ in OTIS-$Q_n$.
The diameter of E-OTIS-$Q_n$ is also $\lfloor {4n+4\over 3} \rfloor$
\cite{D07},  less than the diameter
of OTIS-$Q_n$ by a constant factor.

The OTIS structure is an attractive option for multiprocessor
systems as it offers the benefits from both optical and electrical
technologies. E-OTIS-$Q_n$ retains OTIS-$Q_n$ as a subgraph
and thus has almost all the desirable properties of OTIS-cube
but  contains some additional optical links.
 The advantage gained by adding those extra links, namely,
uniform node degree, reduction in diameter by almost one-third, and
improved fault-tolerance, far outweigh the cost.
 These advantages
make E-OTIS-$Q_n$ a suitable architecture for multiprocessor
interconnection network.

We have quoted the parts of the paper \cite{D07} presenting the algorithm for
shortest path routing and the proof of its optimality,  as they
constitute the basis for our main contribution in this paper. i.e., finding
the eccentricities of the nodes of OTIS-$Q_n$ and E-OTIS-$Q_n$.
In recent times there are a number of research works published on several
graph properties like eccentricity, proximity, remoteness \cite{HP12},
\cite{MWZ12}.

The rest of the paper is organized as follows. In the next section
we present the definitions and basic properties. The next two
sections deal with routing algorithm and eccentricities  of
the nodes in
E-OTIS-$Q_n$ respectively. Section 5 concludes the paper.

\section{Definitions  and Preliminaries}

The binary hypercube $Q_n$ has $N=2^n$ nodes labeled from 0 to $2^n-1$.
Two nodes are connected by an edge if and only if, their labels, in
binary,  differ in exactly one bit position. $Q_n$ has degree
$n$ and diameter $n$. Given two nodes $u$ and $v$ in $Q_n$, we denote
by $H(u,v)$ the {\it Hamming
distance} between $u$ and $v$, i.e., the number of bit positions
in which $u$ and $v$ differ. $H(u,v)$ is also the length of the
shortest path between $u$ and $v$ in $Q_n$.

The OTIS-$Q_n$ is composed of $N=2^n$ node-disjoint subgraphs $Q_n^{(0)}$,
$Q_n^{(1)}$,  $\ldots Q_n^{(N-1)}$, called groups. Each of these groups is
isomorphic to a binary hypercube $Q_n$. A node $<\hspace{-.1cm}g,x\hspace{-.1cm}>$ in OTIS-$Q_n$
corresponds to a node of address $x$ in group $Q_n^{(g)}$. We refer to
$g$ as the group address of node $<\hspace{-.1cm}g,x\hspace{-.1cm}>$ and to $x$ as its processor
address.
\begin{definition} The OTIS-$Q_n$ network is an undirected graph
$(V,E)$ given by

$V=\{<\hspace{-.1cm}g,x\hspace{-.1cm}>|g,x\in Q_n\}$
and\\
$E =\{(<\hspace{-.1cm}g,x  \hspace{-.1cm}>,<\hspace{-.1cm}g,y  \hspace{-.1cm}>) | (x  ,y  )$ is an edge in $Q_n\} \cup
\{(<\hspace{-.1cm}g,x\hspace{-.1cm}>,<\hspace{-.1cm}x,g\hspace{-.1cm}>)|g\ne x$ in $Q_n$\}
\end{definition}

 An intra-group edge of the form $(<\hspace{-.1cm}g,x\hspace{-.1cm}>, <\hspace{-.1cm}g,y\hspace{-.1cm}>)$ corresponds
to an electrical link.
\begin{definition}
The electrical link $(<g,x>, <g,y>)$ is referred
to as link $i$ if $x$ and $y$ differ in the $i^{th}$ bit.
\end{definition}

 An inter-group link of the form $(<\hspace{-.1cm}g,x\hspace{-.1cm}>, <\hspace{-.1cm}x,g\hspace{-.1cm}>)$
corresponds to an optical link.

Some of the properties of OTIS-$Q_n$ are given below \cite{DA02}:
\begin{itemize}

\item  The size of OTIS-$Q_n$ is $2^{2n}$.

\item  The degree of
a node $<\hspace{-.1cm}g,x\hspace{-.1cm}>$ is $n+1$, if $g\ne x$ and
 $n$, if $g=x$.

\item  The distance
between two nodes $<\hspace{-.1cm}g,x\hspace{-.1cm}>$ and $<\hspace{-.1cm}h,y\hspace{-.1cm}>$, denoted as $d(g,x,h,y)$ is given
 as :

$d(g,x,h,y) = \cases{ H(x,y), &if $g=h$\cr
  {\rm min} \cases{H(g,h)+H(x,y)+2 \cr
   H(x,h)+H(g,y)+1},                               &otherwise}$

\item  The diameter of OTIS-$Q_n$ is $(2n+1)$.
\end{itemize}

\begin{definition} The enhanced OTIS-$Q_n$
is obtained from OTIS-$Q_n$ by connecting every node of the form $<\hspace{-.1cm}g,g\hspace{-.1cm}>$ to
the node $<\hspace{-.1cm}\bar g, \bar g\hspace{-.1cm}>$ by an optical link, where $\bar g$ is obtained by complementing
all the bits of $g$.
We refer to  the links of these form as E-links.
\end{definition}

\section{Shortest Path Routing}

The routing algorithms for OTIS-$Q_n$ and E-OTIS-$Q_n$  described
in this section are as given in \cite{D07}.

\subsection{Routing in OTIS-$Q_n$}
The routing algorithm for OTIS-$Q_n$ is taken from
\cite{DA02} but a  few observations are added.
 For a source node $<g,x>$ and destination node $<h,y>$, this
algorithm always gives a path of length $d(g,x,h,y)$ defined in section 2.

If $(g=h)$, then $d(g,x,h,y)=H(x,y)$.
Otherwise, $d(g,x,h,y)$ is minimum $\{l_1,l_2\}$, where
$l_1=H(x,h)+H(g,y)+1$ and
$l_2=H(g,h)+H(x,y)+2$.

The paths corresponding to $l_1$ and $l_2$ are referred to as
{\it path-1} and {\it path-2}
respectively.

The authors made  the following observations.
\\
{\bf Obs 1.} If $g=x$ or $h=y$, then $l_1 <l_2$\\
{\bf Obs 2.} If $h=x$, then $l_1<l_2$ \\
{\bf Obs 3.} If $l_2 <l_1$ then we must have $g\ne x$ and $h \ne x$.

Now  the algorithm {\bf Route} $(g,x,h,y)$ to route from $<g,x>$ to
$<h,y>$ in an OTIS-$Q_n$ is presented.

\begin{algorithm}
 \DontPrintSemicolon
 \SetKwInOut{Input}{input}\SetKwInOut{Output}{output}
 \Input{Nodes $<g,x>$ and $<h,y>$}
 \Output{Path from $<g,x>$ to $<h,y>$}
 \Begin{
  \If {$g=h$}
    { traverse group $g$ to $<g,y>$}
  \Else {
    \If {$H(g,h)+H(x,y)+2 < H(x,h)+H(g,y)+1$}
       {{\bf route2} $(g,x,h,y)$}
     \Else {{\bf route1} $(g,x,h,y)$}
     }
    }
{\bf route1} $(g,x,h,y)$ : $<g,x> \rightarrow <g,h> \rightarrow <h,g> \rightarrow <h,y>$
\\
{\bf route2} $(g,x,h,y)$ : $<g,x> \rightarrow <x,g> \rightarrow <x,h> \rightarrow <h,x> \rightarrow <h,y>$

\caption{Routing in OTIS-cube}
\end{algorithm}

Here {\bf route1} and {\bf route2}
 correspond to {\it path-1} and {\it path-2} respectively.
{\bf route2} is valid only if $g\ne x$ and $h \ne x$. Note
that when  path 2 is followed,  $g\ne x$ and $h\ne x$ by obs. 3.

\subsection{Routing in E-OTIS-$Q_n$}

The E-links which are added to an OTIS-$Q_n$ to get E-OTIS-$Q_n$ are
of the form $(<p,p>, <\bar p, \bar p>)$. It is shown that any shortest
path in E-OTIS-$Q_n$ can contain a link of this form at most once.

\begin{lemma}\label{only1} Any shortest path from a node $<g,x>$ to a node
$<h,y>$ in E-OTIS-$Q_n$, can contain at most one E-link.
 \end{lemma}

\begin{proof} The proof is by contradiction.
Let us suppose a shortest path from $<g,x>$ to $<h,y>$
contains more than one E-link. Suppose the first two
occurrences of such E-links are of the form
  $(<p,p>, <\bar p, \bar p>)$ and
and  $(<q,q>, <\bar q, \bar q>$.

Suppose the path is

{\bf A} : $<g,x> \rightarrow <p,p> \rightarrow < \bar p, \bar p> \rightarrow
<q,q> \rightarrow <\bar q, \bar q>  \rightarrow <h,y>$

The length of
the partial path from $<g,x>$ to $<p,p>$ is $H(g,p)+H(x,p)+1$ , if $g\ne p$
and $H(x,p)$, if $g=p$.
Similarly length of the partial path from $<\bar p, \bar p>$ to
$<q,q>$ is $H(\bar p, q)+H(\bar p, q)+1$.
Hence length of the partial path from $<g,x>$ to $<\bar q, \bar q>$
is greater than or equal to  $H(g,p)+H(p,x)+2H(\bar p,q)+2$

The node  $<\bar q, \bar q>$ can be reached
 from $<g,x>$ without using any E-link
and length of that path is $H(g,\bar q)+H(x,\bar q)+1$, if $g \ne \bar q$
and $H(x, \bar q)$, if $g = \bar q$.

Now $H(g,p)+H(p,x)+2H(\bar p, q)+2$\\
$= H(g,p)+H(p,\bar q) + H(x,p) + H(p,\bar q) +2$
\hfill   [ as $H(\bar p, q) = H(p, \bar q)$ ]\\
$\ge H(g,\bar q)+ H(x, \bar q) + 2$\\
$>H(g,\bar q) + H(x,\bar q) +1$

Hence, path {\bf A} cannot be a shortest path, which is a contradiction.
\end{proof}

\subsubsection{Minimizing the number of electrical links}
The shortest path between $<g,x>$ and
$<h,y>$ either won't have any E-link or will have only one E-link.
In the later case the path will be of the form

{\bf E} : $<g,x> \rightarrow
<b,b> \rightarrow <\bar b, \bar b> \rightarrow <h,y>$

In the later part of this section  whenever a path with E-link is
considered, it is  assumed
that the E-link is of the form $(<b,b>, <\bar b,\bar b>)$.

\begin{definition}
Let $C_i(g,x,h,y)$ be the number of times electrical link $i$ is used
in the path {\bf E}. It is called  the
 cost associated with bit $i$.  We can write $C_i(g,x,h,y)$  as $C_i
(g,x,b,b)+C_i(\bar b,\bar b,h,y)$.
The $i^{th}$ bit of $g$ is denoted as  $g[i]$.
\end{definition}

The following lemmas are regarding $C_i(g,x,b,b)$ and $C_i(\bar b,
\bar b,h,y)$ in the path {\bf E}.

\begin{lemma}\label{l1}
 If $g[i] \ne x[i]$ then $C_i(g,x,b,b)=1$.
\end{lemma}

\begin{proof}If
$x[i] \ne b[i]$, $g[i]=b[i]$ and if $x[i]=b[i]$, $g[i] \ne b[i]$. That is,
The $i^{th}$ bit needs to be changed   either for moving from $<g,x>$ to $<g,b>$
or for moving from $<b,g>$ to $<b,b>$.
\end{proof}

\begin{lemma}\label{l2}
If $g[i]=x[i]$ then $C_i(g,x,b,b)=\cases{0, &if  $g[i]=b[i]$\cr
2, &otherwise}$
\end{lemma}
\begin{proof}If $g[i]=b[i]$, the $i^{th}$ bit  need not
be changed  at all.
If $g[i]\ne b[i]$, the $i^{th}$ bit needs to be changed twice, once for
moving from $<g,x>$ to $<g,b>$ and once for moving from $<b,g>$ to $<b,b>$.
\end{proof}

\begin{lemma}\label{l3} If $h[i] \ne y[i]$ then $C_i(\bar b, \bar b,h,y)=1$.
\end{lemma}

\begin{proof} The proof is similar to that of lemma \ref{l1}.
\end{proof}

\begin{lemma}\label{l4}
If $h[i]=y[i]$ then $C_i(\bar b, \bar b,h,y)
=\cases{ 0, &if  $h[i]= \bar b[i]$ \cr
2, &otherwise}$
\end{lemma}
\begin{proof}The proof is similar to that of lemma \ref{l2}
\end{proof}

The lemmas \ref{l1} to  \ref{l4} are applied and
 $C_i(g,x,h,y)$ is computed  for each of the following cases for $b[i]=g[i]$
and $b[i] \ne g[i]$.
Whenever  $C_i(g,x,h,y)$ is expressed as sum of two terms, the first
term corresponds to $C_i(g,x,b,b)$ and the second term corresponds
to $C_i(\bar b, \bar b, h,y)$.
\\
{\bf Case 1.} $g[i] \ne x[i]$ and $h[i]=y[i]=g[i]$.
\\
Here $C_i(g,x,b,b)=1$ by lemma \ref{l1}.\\
 By lemma \ref{l2},
$C_i(\bar b, \bar b, h,y)=\cases{0,  &if
$b[i] \ne g[i]$ \cr
2,  &otherwise}$.\\
 Hence $C_i(g,x,h,y)=\cases{ 1, &if
$b[i] \ne g[i]$ \cr
  3,  &if $b[i]=g[i]$}$.
\\
{\bf Case 2.} $g[i] \ne x[i]$ and $h[i] =y[i] \ne g[i]$
\\
If $b[i] \ne g[i]$, $C_i(g,x,h,y)=1+2=3$. If
 $b[i]=g[i]$,  $C_i(g,x,h,y)=1+0=1$
\\
{\bf Case 3.} $g[i] \ne x[i]$ and $h[i] \ne y[i]$
\\
Here lemma \ref{l1} and \ref{l3} apply. So,  $C_i(g,x,h,y)=1+1=2$.
\\
{\bf Case 4.} $g[i]=x[i]$ and $h[i]=y[i]=g[i]$
\\
If $b[i]=g[i]$, then $C_i(g,x,h,y)=0+2=2$ by lemma \ref{l2} and
\ref{l4}. If $b[i] \ne g[i]$,
then $C_i(g,x,h,y)=2+0=2$.
\\
{\bf Case 5.} $g[i]=x[i]$ and $h[i]=y[i] \ne g[i]$
\\
Here, lemma \ref{l2} and \ref{l4}  apply. If $b[i]=g[i]$, then $C_i(g,x,h,y)=0+0=0$.
If $b[i]\ne g[i]$, then $C_i(g,x,h,y)=2+2=4$.
\\
{\bf Case 6.}  $g[i]=x[i]$ and $h[i] \ne y[i]$
\\
If $b[i]=g[i]$, then $C_i(g,x,h,y)=0+1=1$. \\
If $b[i] \ne g[i]$, then $C_i(g,x,h,y)=2+1=3$.

 {\bf Finding the optimal $b$:}
The length of the path (B) is equal to the number of
optical links in the path + $\sum C_i(g,x,h,y)$.
 So, for finding the shortest path,  we have to consider how to minimize
$C_i$ as well as the
number of optical links.
 Now, if $b=g$ there is no optical link in the path from
$<g,x>$ to $<b,b>$. But if $b \ne g$, then there is one optical link
in that part of the path. Summarizing case 1 to 6, we find that it
is advantageous to have $b[i] \ne g[i]$ only in case 1. We define
the set $S_g= \{i| g[i]\ne x[i], h[i]=y[i]=g[i]\} $ corresponding
to case 1. If $|S_g|=0$, we can make $b=g$ and avoid having an
optical link in the partial path $<g,x>$ to $<b,b>$. If $|S_g|\ne 0$
we cannot avoid optical link in this part of the path. In that case
we try to avoid optical link in the part of the path from $<\bar b, \bar b>$
to $<h,y>$. We define $S_h= \{i | h[i] \ne y[i], g[i]=x[i]=h[i]\}$.
By similar logic $S_h$ is the set of bit positions $i$ where
it is advantageous to have $\bar b[i] \ne h[i]$. If $|S_h|=0$, we can set
$\bar b=h$ and avoid optical link in the partial path from
$<\bar b, \bar b>$.

\subsubsection{routing algorithm}
Now  the algorithm {\bf RTE} to obtain the
length of the shortest path involving an E-link is presented.  {\bf RTE} returns
two values : the length of the path and the value $b$.

 Let $\oplus$ denote the bitwise exclusive-OR (XOR) operation,
$\odot$ denote the bitwise equivalence operation (XNOR),
 and \& denote the bitwise
AND operation. Let $t$ and $u$ be two $n$-bit binary numbers.

\begin{algorithm}
 \DontPrintSemicolon
 \SetKwInOut{Input}{input}\SetKwInOut{Output}{output}
 \Input{Nodes $<g,x>$ and $<h,y>$}
 \Output{Length of the shortest path involving an E-link, and $b$ where $(<b,b>,<\bar b, \bar b>)$ is the E-link}
  \Begin{
  $t \leftarrow (g\oplus x) \&(g \odot h) \& (h \odot y)$.\;
$u \leftarrow (h \oplus y) \&( g \odot h) \& (g \odot x)$.\;
\If {$t =0$} {$b \leftarrow g$}
\Else {
\If {($u\ne 0$)}   {$b  \leftarrow g \oplus t$}
\Else
 {$b = \bar h$}
 }
return $( d(g,x,b,b)+1+ d(\bar b, \bar b, h,y), b)$
}
\caption{Algorithm RTE}
\end{algorithm}

The terms $d(g,x,b,b)$ and $d(\bar b, \bar b, h,y)$ correspond to the length of the paths $<g,x> \rightarrow <b,b>$
and $<b,b> \rightarrow <h,y>$ respectively
and 1 is due to the use of an E-link $<b,b>$ to $<\bar b, \bar b>$.

\begin{lemma} The algorithm {\bf RTE} always gives the shortest path
involving an E-link. \end{lemma}

\begin{proof} By lemma \ref {only1}, a shortest path involving
an E-link, cannot have more than one E-links. The algorithm RTE
uses only one  E-link $<b,b>$ to $<\bar b, \bar b>$. Thus, we
only need to show that choice of $b$ is an optimal one.
Clearly, $S_g =\{ i| t[i]=1\}$ and $S_h=\{i |u[i]=1\}$.
If $t=0$, then $|S_g|=0$, and we can make $b=g$, avoiding
optical link while going from $<g,x>$ to $<b,b>$.
If $t\ne 0$, then $|S_g|>0$. For
every $i \in S_g$, if we make $b[i] \ne g[i]$ the cost associated
with this bit is 1, compared to cost 3 when $b[i]$ is equal to
$g[i]$. So, it is advantageous to have $b[i] \ne g[i]$ for every
$i \in S_g$. That is achieved by making $b = g \oplus t$. Now
we cannot avoid optical link in route $(g,x,b,b)$. In this case,
we check if we can avoid optical link in route $(\bar b, \bar b, h,y)$,
which will be possible if $\bar b=h$ i.e., $|S_h|=0$. If $u=0$,
i.e., $|S_h|=0$, we set $b =\bar h$ avoiding optical link in
the partial path  $<\bar b, \bar b>$  to $<h,y>$. Hence, the choice of $b$ is optimal
considering both the cost $C_i$ for each bit position $i$ and
use of optical links.
\end{proof}

Now  the algorithm for routing in E-OTIS-$Q_n$ which is
called {\bf Eroute} is presented.

\begin{algorithm}
 \DontPrintSemicolon
 \SetKwInOut{Input}{input}\SetKwInOut{Output}{output}
 \Input{Nodes $<g,x>$ and $<h,y>$}
 \Output{Shortest path Routing in E-OTIS-cube}
\Begin{
$L_1 \leftarrow d(g,x,h,y)$.\;
$(L_2,b)  \leftarrow {\bf RTE} (g,x,h,y)$. \;
\If {$L_1 \le L_2$} { {\bf route} $(g,x,h,y)$}
\Else  {
	{\bf route} $(g,x,b,b)$.\;
	follow E-link to $< \bar b , \bar b>$.\;
	{\bf route} $(\bar b, \bar b, h,y)$.\;
}
}
\caption{Algorithm {\bf Eroute} }
\end{algorithm}

\begin{theorem} The algorithm {\bf Eroute} always routes
by the shortest path. \end{theorem}

\begin{proof} The shortest path between two nodes $<g,x>$ and $<h,y>$
either does not use an E-link or uses a single E-link. In the
former case, the length of the shortest path is $L_1$ and routing
is by algorithm {\bf route}. In the later case, the length of
the shortest path is $L_2$ as found by {\bf RTE}.
\end{proof}

 The number of
optical links in the computed shortest path is at most 3.

\section {Eccentricity of nodes of E-OTIS-$Q_n$ and its Diameter}

We consider a source $<g,x>$, a destination $<h,y>$ and a set of
three paths between them. We omit the case when $g=h$, since in that
case the distance between $<g,x>$ and $<h,y>$ is at most $n$.
\\
path 1 : $<g,x> \rightarrow <g,h> \rightarrow <h,g> \rightarrow <h,y>$\\
path 2 : $<g,x> \rightarrow <x,g> \rightarrow <x,h> \rightarrow <h,x>
\rightarrow <h,y>$\\
path3 : $<g,x> \rightarrow <b,b> \rightarrow <\bar b,\bar b> \rightarrow
<h,y>$ where choice of $b$ is optimal (as found by {\bf RTE}).

Here, path 1 uses one optical link, path 2 uses two
optical links, path 3 uses one E-link and at most two other
optical links.
Let $l_1, l_2$ and $l_3$ be the lengths of path 1, path 2, and
path 3 respectively.  The shortest path between $<g,x>$ and $<h,y>$ is
of length $l=  min(l_1,l_2,l_3)$. Let $A_i, B_i$ and $C_i$ be the
number of times link $i$ is used in path 1, path 2 and path 3 respectively.
We find the value of $A_i, B_i$ and $C_i$ for each of the following
cases.
\\
\\
{\bf Case 1.} $g[i] \ne x[i]$, $h[i]=y[i]=g[i]$ : Consider path
1. Since $x[i] \ne h[i]$ and
$g[i]=y[i]$, we need to use link $i$ while going from
$<g,x>$ to $<g,h>$ but not while going from $<h,g>$ to
$<h,y>$. Hence, $A_i=1$. For path 2, we do not need
to use link $i$ while going from $<x,g>$ to $<x,h>$ but
use it while going from $<h,x>$ to $<h,y>$. Thus, $B_i=1$.
For path 3, as explained in  section 3,  $C_i=1$.
\\
{\bf Case 2.} $g[i] \ne x[i]$, $h[i]=y[i] \ne g[i]$ : $A_i=1$, $B_i=1$,
$C_i=1$
\\
{\bf Case 3.} $g[i]=x[i]=h[i]$, $h[i] \ne y[i]$ :
$A_i=1$, $B_i=1$, $C_i=1$
\\
{\bf Case 4.} $g[i]=x[i] \ne h[i]$, $h[i] \ne y[i]$ :
$A_i=1$, $B_i=1$, $C_i=1$
\\
{\bf Case 5.} $g[i]=x[i]$, $h[i]=y[i]=g[i]$ :
$A_i=0$, $B_i=0$, $C_i=2$
\\
{\bf Case 6.} $g[i]=x[i]$, $h[i]=y[i] \ne g[i]$ :
$A_i=2$, $B_i=2$, $C_i=0$
\\
{\bf Case 7.} $g[i] \ne x[i]$, $h[i] \ne y[i]$, $h[i]=g[i]$ :
$A_i=2$, $B_i=0$, $C_i=2$
\\
{\bf Case 8.} $g[i] \ne x[i]$, $h[i] \ne y[i]$, $h[i] \ne g[i]$ :
$A_i=0$, $B_i=2$, $C_i=2$

Define a set of bit positions
 $S_1 = \{i, g[i]\ne x[i], h[i]=y[i]=g[i]\}$. Clearly, $S_1$ corresponds
to the case 1. Similarly we define $S_2$ to $S_8$, where $S_i$ corresponds
to case $i$. Since, each bit will belong to exactly one of $S_i$'s,
  we have $\sum_i |S_i| =n$. Also $H(g,x)=|S_1|+|S_2|+|S_7|+|S_8|$. Note that $S_1$ is identical
to
the set $S_g$ defined in previous section. \\
Let $T=|S_1|+|S_2|+|S_3|+|S_4|$.
Now we state the following lemmas.

\begin{lemma} \label{sl1}
$T+2|S_5|+2|S_7|+2|S_8|+1 \le l_3 \le T+2|S_5|+2|S_7|+2|S_8|+3$
\end{lemma}

\begin{proof}$l_3=\sum_i C_i +$ number of optical links in the path.\\
Since, $C_i=1$ if $i \in S_1 \cup S_2\cup S_3\cup S_4$, \\
$C_i=2$ if $i \in S_5 \cup S_7 \cup  S_8$, \\
and $C_i=0$ if $i \in S_6$.\\
Also, the path uses at least one and at most three optical links. Hence,
the proof follows.
\end{proof}

\begin{corollary}\label{actual} If $S_1=0$, and $S_3+S_5+S_7>0$, the length of
the shortest path involving an E-link is $l_3=T+2|S_5|+2|S_7|+2|S_8|$.
\end{corollary}
\begin{proof}
If $S_1=0$, we can set $b=g$ (discussion subsection 3.2.1) and do not need optical link while going from
$<g,x>$ to $<g,g>$. Also $S_3+S_5+S_7>0$ implies that $\bar g \ne h$. Hence we need one optical link in the path
from $<\bar g, \bar g>$ to $<h,y>$.
\end{proof}

\begin{lemma}\label{sl2} $l_1 = T + 2|S_6|+2 |S_7| +1$ \end{lemma}

\begin{lemma}\label{sl3} $l_2 = T + 2|S_6|+ 2|S_8|+2$ \end{lemma}

\subsection{Eccentricity of nodes in OTIS-cube and Enhanced OTIS-cube}
A node in OTIS-$Q_n$ or E-OTIS-$Q_n$ is represented as $<g,x>$ where $g$ is
the group number, and $x$ is the node number within the group. We classify the
nodes based on the value of $H(g,x)$. All the nodes belong to one of the
$n+1$ classes corresponding to the value $H(g,x)=0,1, \ldots , n$. We show that
nodes belonging to the same class have the same eccentricity. We show
that for OTIS-$Q_n$, the eccentricity is maximum for
node with $H(g,x)=0$,
then it decrease
with increase in value of $H(g,x)$ and is minimum for $H(g,x)=n$.
For E-OTIS-$Q_n$,
the eccentricity is minimum for $H(g,x)=0$,
increases with increase in $H(g,x)$  upto $\lfloor 2n/3 \rfloor$,
and then again decreases and is minimum for $H(g,x)=n$.

The following lemmas (lemma \ref{ggbar} and lemma \ref{gx})  hold for
OTIS-$Q_n$ as well
as E-OTIS-$Q_n$.

\begin{lemma}\label{ggbar} A node $<g, \bar g>$ has eccentricity at
most $(n+1)$.
\end{lemma}

\begin{proof}We consider two paths between
$<g,\bar g>$ and $<h,y>$. Let $l_i$ denote the length of path $i$.\\
\\
path 1: $<g,\bar g> \rightarrow  <g,h> \rightarrow <h,g> \rightarrow <h,y>$ :\\
 $l_1=H(\bar g,h)+1+H(g,y)$\\
path 2: $<g,\bar g> \rightarrow < \bar g,  g> \rightarrow <\bar g, h>
\rightarrow <h, \bar g> \rightarrow <h,y>$ :\\
$l_2= 1+ H(g, h) + 1+ H (\bar g, y)$

Let $l_{min}= {\rm minimum} (l_1,l_2)$. Then $l_{min} \le (l_1+l_2)/2$.
Substituting $H(\bar g, h)=n-H(g,h)$, we get $l_{min} \le (2n+3)/2$. Since
$l_{min}$ must be an integer,  $l_{min} \le \lfloor (2n+3)/2 \rfloor = n+1$.
\end{proof}

\begin{lemma}\label{gx} A node $<g,x>$,  has
eccentricity at most
$(2n+1-H(g,x))$.
\end{lemma}

\begin{proof}
Since, we can
reach $<g,\bar g>$ from $<g,x>$ in $n-H(g,x)$ steps, applying
lemma \ref{ggbar} we have a path  from $<g,x>$ to $<h,y>$ of length
$(n-H(g,x)+ n+1)=2n+1-H(g,x)$.  \end{proof}

\begin{lemma}\label{eco} In an OTIS-$Q_n$, the eccentricity of a node
$<g,x>$ is equal to $2n+1-H(g,x)$.
\end{lemma}

\begin{proof}Consider a node $p=<g,x>$ with $H(g,x)=A$.
 From lemma \ref{gx}, the eccentricity of this node  is
less than or equal to $(2n+1-A)$.

We consider two cases.

{\bf Case 1:} $A$ is even. Let $q=<h,y>$ be another node such that
for the pair $(p, q)$, $|S_6|=n-A$, $|S_7|=A/2$ and $|S_8|=A/2$.
Then by
lemma \ref{sl2}, and \ref{sl3},  we have

$l_1=2|S_6|+2|S_7|+1=2(n-A)+2(A/2)+1=2n-A+1$

$l_2=2|S_6|+2|S_8|+2=2(n-A)+2(A/2)+2=2n-A+2$

Here, $l_1<l_2$, and
hence, the eccentricity of
node $p$ is $(2n+1-A)$.

{\bf Case 2:} $A$ is odd. Let $q$ be such that $|S_7|=(A+1)/2$, $|S_8|=(A-1)/2$ and
$|S_6|=n-A$. Proceeding as before, $l_1=2n-A+2$, and  $l_2=2n-A+1$.

Here, $l_2<l_1$ and hence, eccentricity of node $p$ is
$(2n+1-A)$.  \end{proof}

\begin{lemma}\label{ec1} In an E-OTIS-$Q_n$, for $H(g,x)>\lfloor 2n/3\rfloor$, the eccentricity of a node $<g,x>$
is equal to $(2n+1-H(g,x))$.
\end{lemma}

\begin{proof}Consider a node $p=<g,x>$ with $H(g,x)=A$ and
$A>\lfloor 2n/3\rfloor $. From lemma \ref{gx}, the eccentricity of this node  is
less than or equal to $(2n+1-A)$.

We consider two cases.

{\bf Case 1:} $A$ is even. Let $q=<h,y>$ be another node such that
for the pair $(p,q)$, $|S_6|=n-A$, $|S_7|=A/2$ and $|S_8|=A/2$.
Then by
lemma \ref{sl2}, \ref{sl3}, and corollary \ref{actual}, we have

$l_1=2|S_6|+2|S_7|+1=2(n-A)+2(A/2)+1=2n-A+1$

$l_2=2|S_6|+2|S_8|+2=2(n-A)+2(A/2)+2=2n-A+2$

$l_3=2|S_5|+2|S_7|+2|S_8|+2=2(A/2)+2(A/2)+2=2A+2$

Here, $l_1<l_2$ and as $A>\lfloor 2n/3\rfloor $, $l_1<l_3$. Hence, eccentricity of
node $p$ is $(2n+1-A)$.

{\bf Case 2:} $A$ is odd. Let $q$ be such that $|S_7|=(A+1)/2$, $|S_8|=(A-1)/2$ and
$|S_6|=n-A$. Proceeding as before, $l_1=2n-A+2$, $l_2=2n-A+1$ and $l_3=2A+2$.

Here, $l_2<l_1$ and $l_2<l_3$. Hence, eccentricity of node $p$ is
$(2n+1-A)$.  \end{proof}

\begin{lemma}\label{gx2} In an E-OTIS-$Q_n$, the eccentricity of a node $<g,x>$
is at most   $n+\lfloor {H(g,x)+3\over 2}\rfloor$.
\end{lemma}

\begin{proof}Consider a node $p=<g,x>$ with $H(g,x)=A$ and $A<2n/3$.
Consider another node $q=<h,y>$.
Then by
lemma \ref {sl2}, \ref{sl3}, and \ref{sl1}, we have,

$l_3 \le T+2|S_5|+2|S_7|+2|S_8|+3$

$l_1 = T + 2|S_6|+2 |S_7| +1$

$l_2 = T + 2|S_6|+ 2|S_8|+2$

We consider two cases.

{\bf Case 1:} $|S_7|\le |S_8|$. Here $l_1<l_2$. Hence, path between $p$ and
$q$ is of length $l \le (l_1+l_3)/2$.

$l_1+l_3= 2(T+|S_5|+|S_6|+|S_7|+|S_8|)+2|S_7|+4$\\
$\Rightarrow l_1+l_3 \le 2n+|S_7|+|S_8|+4$

If $|S_1|=0$, then $l_3$ is reduced by 1 (corollary \ref{actual}) and $l_1+l_3 \le 2n+|S_7|+|S_8|+3$.
As $|S_7|+|S_8|\le A$, $l_1+l_3\le 2n+A+3$

If $|S_1|>0$, $|S_7|+|S_8|\le A-1$ and $l_1+l_3\le 2n+A+3$

{\bf Case 2:} $|S_7|>|S_8|$. Here, $l_2<l_1$.
Path between $p$ and $q$ is of length $l\le (l_2+l_3)/2$.

$l_2+l_3=2n+2|S_8|+5 \le 2n+|S_8|+(|S_7|-1)+5=2n+|S_8|+|S_7|+4$.
Proceeding as in Case 1, we can show that $l_2+l_3\le 2n+A+3$.

Combining cases 1 and 2, we conclude that eccentricity of $p$ is
at most $n+\lfloor {A+3 \over 2}\rfloor$.

\end{proof}

\begin{lemma}\label{ec2} In an E-OTIS-$Q_n$, for $H(g,x)\le \lfloor 2n/3\rfloor$, the eccentricity of a node $<g,x>$
is equal to $n+\lfloor {H(g,x)+3\over 2}\rfloor$.
\end{lemma}

\begin{proof}Consider a node $p=<g,x>$ with $H(g,x)=A$ and
$A\le 2n/3$. From lemma \ref{gx2}, the eccentricity of this node  is
less than or equal to $n+\lfloor{A+3 \over 2}\rfloor$.

To show that eccentricity of node $p$ is equal to $n+\lfloor{A+3\over 2}
\rfloor$ we consider the following cases.

{\bf Case 1:} $A$ is even : Consider two sub cases

{\bf Case 1a:} $(2n-A) = 0 \bmod 4$. Let us take $q$ such that
             $|S_7|=A/2$, $|S_8|=A/2$, $|S_5|=(2n-3A)/4$, and $|S_6|=(2n-A)/4$.

Then by
corollary \ref{actual} and lemma \ref {sl2}, \ref{sl3}, we have,
$l_3=2(2n-3A)/4+A+A+2=n+A/2+2$\\
$l_1=2(2n-A)/4+A+1=n+A/2+1$ and $l_2>l_1$.

Hence minimum of $l_1$, $l_2$ and $l_3$ is $n+(A+2)/2$.
For even $A$, $n+(A+2)/2=n+\lfloor{A+3\over 2}\rfloor$

{\bf Case 1b:} $(2n-A)=2 \bmod 4$. Let
             $|S_7|=A/2$, $|S_8|=A/2$, $|S_5|=\lfloor {2n-3A\over 4}
\rfloor$ = $(2n-3A-2)/4$, and $|S_6|=(2n-A+2)/4$.

Then by
corollary \ref{actual} and lemma \ref {sl2}, \ref{sl3}, we have,
$l_3=2(2n-3A-2)/4+A+A+2=n+A/2+1$\\
$l_1=2(2n-A+2)/4+A+1=n+A/2+2$ and $l_2>l_1$.

Hence minimum of $l_1$, $l_2$ and $l_3$ is $n+(A+2)/2$.
For even $A$, $n+(A+2)/2=n+\lfloor{A+3\over 2}\rfloor$

{\bf Case 2:} $A$ is odd.

{\bf Case 2a:} $(2n-A)=1 \bmod 4$. Let
            $S_3=1$,  $|S_7|=(A+1)/2$, $|S_8|=(A-1)/2$, $|S_5|=\lfloor {2n-3A\over 4}
\rfloor$ = $(2n-3A-3)/4$, and  $|S_6|=(2n-A-1)/4$.

Then by corollary \ref{actual} and
lemma \ref {sl2}, \ref{sl3}, we have,
$l_3=1+ 2(2n-3A-3)/4+(A-1)+(A+1)+2=n+(A+3)/2$\\
$l_2= 1+ 2(2n-A-1)/4+2(A-1)/2+2=n+(A+3)/2$ and $l_2<l_1$.

{\bf Case 2b:} $(2n-A)=3 \bmod 4$. Let
             $|S_7|=(A+1)/2$, $|S_8|=(A-1)/2$, $|S_5|=\lfloor {2n-3A\over 4}
\rfloor$ = $(2n-3A-1)/4$, and  $|S_6|=(2n-A+1)/4$.

Then by corollary \ref{actual} and
lemma \ref {sl2}, \ref{sl3}, we have,
$l_3= 2(2n-3A-1)/4+(A-1)+(A+1)+2=n+(A+3)/2$\\
$l_2= 2(2n-A+1)/4+2(A-1)/2+2=n+(A+3)/2$ and $l_2<l_1$.

Hence, combining Cases 1 and 2, we prove that eccentricity of node
$<g,x>$ is $n + \lfloor {H(g,x)+3 \over 2}\rfloor$ for $H(g,x) \le \lfloor
2n/3 \rfloor$.
\end{proof}

Table 1 compares the eccentricities of an OTIS-$Q_8$ and an E-OTIS-$Q_8$ for
different values of $H(g,x)$.\\
{\bf Remark:} The variation in eccentricity of the nodes is less in E-OTIS-$Q_n$ compared to OTIS-$Q_n$.
For $H(g,x)$ close to $n$, the eccentricity values are equal, and for $H(g,x)$ close to zero, the eccentricity
of the nodes in E-OTIS-$Q_n$ is  almost half of that in OTIS-$Q_n$.

\vspace*{13pt}
\centerline{\footnotesize Table~1.  Eccentricity in OTIS-$Q_8$ and E-OTIS-$Q_8$}
\vspace*{-2pt}
\noindent
\begin{center}
{\footnotesize
\begin{tabular}{|c|c|c|}
\hline
$H(g,x)$  &eccentricity in OTIS &eccentricity in E-OTIS\\
\hline
0  &17 &9\\
1  &16 &10\\
2  &15 &10\\
3  &14 &11\\
4  &13 &11\\
5  &12 &12\\
6  &11 &11\\
7  &10 &10\\
8  &9  &9\\
\hline
\end{tabular}
}\end{center}

\subsection{Average eccentricity of OTIS-$Q_n$ and E-OTIS-$Q_n$}

The following lemma establishes average eccentricity of OTIS-$Q_n$.
\begin{lemma}
Average eccentricty of OTIS-$Q_n$ is $(3n+2)/2$
\end{lemma}

\begin{proof}

The eccentricity of a node $<g,x>$ depends only on $H(g,x)$ irrespective of the
value of $g$. So, to find the average we can consider only the nodes
of one particular group. For a given $g$,
the number of nodes $<g,x>$ with $H(g,x)=k$ is ${n \choose k}$ and
they are of eccentricity $2n+1-k$.

Hence sum of the eccentricities of all the nodes in a group is given by
\begin{equation}
\begin{array}{lll}
\sum_{k=0}^n ({n \choose k}(2n+1-k)
&= (2n+1)2^n -n \sum_{k=1}^n {n-1 \choose k-1}\\
&=(2n+1)2^n-n2^{n-1}=(3n+2)2^{n-1}\\
\end{array}
\end{equation}

Dividing the sum by total number of nodes in a group i.e,  $2^n$, we get the
average eccentricity  as $(3n+2)/2$.

\end{proof}

Similarly to find the average eccentricity of E-OTIS-$Q_n$, we
note that for $H(g,x)<\lfloor 2n/3\rfloor$, eccentricity of $<g,x>$ is
$n+ \lfloor {H(g,x)+3 \over 2}\rfloor$ and for $H(g,x) >2n/3$, the
eccentricity is $2n+1-H(g,x)$.

Hence average eccentricity is equal to
\begin{equation}
{1 \over 2^n}
(\sum_{k=0}^{\lfloor 2n/3\rfloor }  {n\choose k}
(n+ \lfloor {k+3 \over 2}\rfloor) +
\sum_{k= \lfloor 2n/3\rfloor +1}^n
  {n \choose k} (2n+1-k))
\end{equation}

It is not possible to give a simplified expression for eccentricity in this
case.
We evaluated eccentricity of E-OTIS-$Q_n$ for different values of $n$
and put them
on a table below. We also put the value $(3n+2)/2$ obtained as average
eccentricity of OTIS-$Q_n$.

\subsection{Diameter of E-OTIS-$Q_n$}
We can now use the eccentricity values to find the diameter of E-OTIS-$Q_n$.
\begin{theorem} Diameter of E-OTIS-$Q_n$ is equal to $\lfloor {4n+4\over 3}
\rfloor$.
\end{theorem}
\begin{proof}Given a node $<g,x>$, in E-OTIS-$Q_n$, its eccentricity is $2n+1-H(g,x)$ for $H(g,x)>\lfloor {2n\over 3}\rfloor$ by
lemma \ref{ec1}. So, for
$H(g,x)$ ranging from $n$ to $\lfloor {2n/3 }\rfloor +1$ eccentricity ranges from $n+1$ to $2n-\lfloor 2n/3 \rfloor$.
Again, by lemma \ref{ec2} for $H(g,x)\le \lfloor {2n\over 3}\rfloor$, the eccentricity of a node $<g,x>$ is equal to
 $n + \lfloor {H(g,x)+3 \over 2}\rfloor$. For $H(g,x)$ ranging from 0 to $\lfloor{2n\over 3}\rfloor$ the eccentricity
of a node $<g,x>$ ranges from $n+1$ to $\lfloor {4n+4\over 3}\rfloor$. Hence, the proof. \end{proof}

The smaller diameter and eccentricities can have important application in implementing parallel algorithm
on E-OTIS-$Q_n$. Since E-OTIS-$Q_n$ has OTIS-$Q_n$ as subgraph, all the algorithms for OTIS-$Q_n$ can be mapped
to E-OTIS-$Q_n$ without change. In those algorithm whenever there is a need for routing between nodes, the
shortest path routing algorithm developed here can be applied and if there is a need for single node broadcast,
the node having lower eccentricity can be chosen as the originator for broadcast
and thus reducing the time
for broadcast.

\vspace*{13pt}
\centerline{\footnotesize Table~2. Average Eccentricity}
\vspace*{-2pt}
\noindent
\begin{center}
{\footnotesize
\begin{tabular}{|c|c|c|}
\hline
$n$ &Average Eccentricity of  &Average Eccentricity of\\
   &OTIS-$Q_n$ &E-OTIS-$Q_n$\\
\hline
4 &7.0 		&5.875\\
5 &8.5 		&7.250\\
6 &10.0           &8.516\\
7 &11.5		&9.695\\
8 &13.0           &11.031\\
9 &14.5         &12.297\\
10 &16.0          &13.501\\
\hline
\end{tabular}
}\end{center}

\section{Conclusion }

For graphs which are node symmetric the eccentricity of all the nodes
are same and equal to its diameter. But for graphs which are not node-symmetric,
the eccentricity of the nodes can vary. We observer that
minimum eccentricity of OTIS-$Q_n$ and E-OTIS-$Q_n$ is
almost half of the diameter of OTIS-$Q_n$s.
In this work, we make a classification
of the nodes in OTIS-$Q_n$ and E-OTIS-$Q_n$
based on their position within a group. It is
shown that
 the nodes belonging to
the same class have the same eccentricity.
Specifically, we have shown that the eccentricity of a node $<g,x>$, in
an  OTIS-$Q_n$ is $2n+1-H(g,x)$. For E-OTIS-$Q_n$, the eccentricity of
a node $<g,x>$ is equal to $n+\lfloor {H(g,x)+3\over 2}\rfloor$, when
$H(g,x)\le \lfloor 2n/3\rfloor$, and $2n+1-H(g,x)$, when
$H(g,x) >\lfloor 2n/3 \rfloor$.
Similar classification can be attempted on other OTIS networks.

\end{document}